\documentclass{article}
\usepackage{wright}
\pdfoutput=1

\AtEveryBibitem{ 
    \clearfield{day}
    \clearfield{month}
    \clearfield{series}
    \clearfield{venue}
    \clearname{editor}
    \clearlist{publisher}
    \clearlist{location} 
    \clearfield{venue}
    \clearfield{issn}
    \clearfield{isbn}
    \clearfield{urldate}
    \clearfield{eventdate}
    \clearfield{pages}
    \clearfield{number}
    \clearfield{volume}
}

\title{\texorpdfstring{\vspace{-1.0cm}}{}Amplitude amplification and estimation require inverses}
\author{
    Ewin Tang\thanks{UC Berkeley. \texttt{\{ewin,jswright\}@berkeley.edu}} 
    \and John Wright\footnotemark[1]
}

\date{}

\begin{document}

\maketitle

\begin{abstract}
We prove that the generic quantum speedups for brute-force search and counting only hold when the process we apply them to can be efficiently inverted.
The algorithms speeding up these problems, amplitude amplification and amplitude estimation, assume the ability to apply a state preparation unitary $U$ and its inverse $U^\dagger$; we give problem instances based on trace estimation where no algorithm which uses only $U$ beats the naive, quadratically slower approach.
Our proof of this is simple and goes through the compressed oracle method introduced by Zhandry~\cite{zhandry19}.
Since these two subroutines are responsible for the ubiquity of the quadratic ``Grover'' speedup in quantum algorithms, our result explains why such speedups are far harder to come by in the settings of quantum learning, metrology, and sensing.
In these settings, $U$ models the evolution of an experimental system, so implementing $U^\dagger$ can be much harder---tantamount to reversing time within the system.
Our result suggests a dichotomy: without inverse access, quantum speedups are scarce; with it, quantum speedups abound.
\end{abstract}
\hypersetup{linktocpage}
\tableofcontents

\section{Introduction}

Perhaps the most ubiquitous subroutines in quantum computing are amplitude amplification and amplitude estimation, introduced by Brassard, H{\o}yer, Mosca, and Tapp~\cite{BHMT02}.
These algorithms generalize the quadratic quantum speedup of Grover's algorithm~\cite{grover96} to such a wide extent that, these days, it's unusual when we \emph{cannot} use them to speed up a classical computation.

Amplitude amplification and estimation perform two common tactics in algorithm design: brute-force search and counting.
Suppose we have a quantum operation which prepares a state that places some mass on a desired answer.
Can we produce the answer?
And can we estimate the amount of mass on that answer?

\begin{problem}[Amplitude amplification and estimation] \label{def:aa-ae}
    Let $X \in \C^{d \times d}$ be a unitary matrix such that
    \begin{align*}
        X\ket{0} = a\ket{\phi_\good}\ket{0} + \sqrt{1 - a^2}\ket{\phi_\bad} \ket{1},
    \end{align*}
    where $\ket{\phi_\good}$ and $\ket{\phi_\bad}$ are quantum states (in particular, $\braket{\phi_\good}{\phi_\good} = \braket{\phi_\bad}{\phi_\bad} = 1$) and $0 \leq a \leq 1$.
    \begin{enumerate}
        \item The goal of \emph{amplitude amplification} is to output a state $0.01$-close to $\ket{\phi_\good}$ in trace distance.
        \item The goal of \emph{amplitude estimation} is to output an $\wh{a}$ such that $a - \eps < \wh{a} < a + \eps$ with probability $\geq 0.99$.
    \end{enumerate}
\end{problem}

The naive way to solve these problems is to repeatedly prepare a copy of $X\ket{0}$ and measure its second register in the computational basis.
The measurement outcome $\ket{0}$ occurs with probability $a^2$, in which case the first register has the desired state $\ket{\phi_\good}$.
So, to solve amplitude amplification, we can wait until we see the outcome $\ket{0}$ and output the corresponding state.
To solve amplitude estimation, we use the measurement outcomes to estimate the $a^2$ probability of seeing $\ket{0}$.
These protocols use $\bigTheta{1/a^2}$ and $\bigTheta{1/\eps^2}$ applications of $X$, respectively.
As mentioned, though, we can do better.

\begin{theorem}[{Algorithm for amplitude amplification, \cite{BHMT02}}] \label{thm:upper-aa}
    Amplitude amplification, as defined in \cref{def:aa-ae}, can be solved with $\bigO{1/a}$ applications of $X$ and $X^\dagger$.
    This is tight.
\end{theorem}

\begin{theorem}[{Algorithm for amplitude estimation, \cite{BHMT02,ar20}}] \label{thm:upper-ae}
    Amplitude estimation as defined in \cref{def:aa-ae} can be solved with $\bigO{1/\eps}$ applications of $X$ and $X^\dagger$.
    This is tight.\footnote{
        The algorithm for amplitude estimation given in \cite{BHMT02} requires controlled access to $X$ to perform phase estimation.
        The exposition here follows the work of Aaronson and Rall, which does not require this controlled access~\cite{ar20} (see also \cite{ggzw21,fhiz23}).
        That work technically assumes that $\eps < a$, but the $\eps \geq a$ can easily be handled by the same set of techniques.
    }
\end{theorem}

These routines, along with their more sophisticated descendants, are pillars in the field of quantum algorithms.
Amplitude amplification appears when we want to amplify the success probability of a protocol.
It is a central source of speedup in combinatorial settings, including problems involving constraint satisfaction~\cite{ambainis04}, strings~\cite{rv03,aj23}, and graphs~\cite{dhhm06,abikpv19}.
Enhanced ``oblivious''~\cite{bccks17}, ``variable-time''~\cite{ambainis12}, and QSVT-based~\cite{gslw18} forms of amplification are commonly applied to linear algebraic tasks like ground state preparation~\cite{pw09,lt20}, Hamiltonian simulation, and solving linear systems.
Amplitude estimation appears when we want to estimate something, from means of distributions~\cite{ko23} to physical quantities~\cite{kos07,rall20} to entire quantum states~\cite{vAcgn22}.
Taken altogether, these two subroutines explain why quadratic quantum speedups are so plentiful that they seem almost mundane.

Amplitude amplification and estimation have a curious feature: they require the ability to apply both the quantum operation $X$ and its inverse $X^\dagger$.
We call these \emph{forward} and \emph{inverse} queries, respectively.
But the naive algorithms with complexity $\bigTheta{1/a^2}$ and $\bigTheta{1/\eps^2}$ don't need access to the inverse.
So why do the improved algorithms need it?
And is it necessary for this quadratic speedup?

We can identify from inspecting these algorithms that they do not make sense without inverse queries.
They follow the same principle as Grover's algorithm.
Using applications of $X$ and $X^\dagger$, we can construct quantum circuits for reflections about the good state, $\ket{\phi_\good}\ket{0}$, and about the initial state, $X\ket{0}$.
By alternating these reflections, we can rotate from the initial state towards the good state at a rate which is linear in $a$.
These reflections facilitate an elegant analysis: the quantum state lies in a two-dimensional subspace of the full Hilbert space throughout the algorithm, so the algorithm admits a purely geometric interpretation.
It uses inverse queries to implement the reflection about the initial state, $X (I - 2\proj{0}) X^\dagger$.
Without access to $X^\dagger$, these reflections may not be efficiently implementable, breaking the algorithm.
It moreover seems inevitable that, instead of being confined in a two-dimensional subspace, an algorithm which uses only forward queries must careen off into exponentially large Hilbert space.\footnote{
    The usage of $X^\dagger$ here is related to the idea of \emph{uncomputation}: applying $X$ introduces garbage which interferes with the algorithm's intended behavior, and $X^\dagger$ is precisely the operation needed to remove the garbage.
}

However, it's hard to believe that inverse queries are genuinely necessary for the quadratic speedup.
After all, for tasks in quantum learning, we can frequently achieve $\bigO{1/\eps}$-type quadratic speedups without access to the inverse.
Hamiltonian learning from real-time evolutions~\cite{htfs23}; learning a unitary's spectrum~\cite[Appendix B]{hkot23}; and even full tomography of the unitary~\cite{hkot23} all attain a $\bigO{1/\eps}$ scaling with only forward queries.\footnote{
    For a simple example of this speedup, consider the following task.
    Let $X$ be either the unitary $U$ or $V$, where the diamond distance between $U$ and $V$ as quantum channels is $\eps$.
    Then we can distinguish which with $\bigTheta{1/\eps}$ forward queries by applying $(XU^\dagger)^{0.1/\eps}$ to an appropriately chosen input state: if $X = U$, then this will not change the input, and if $X = V$, the final state will be distinguishable from the input state with constant probability.
}
Further, if we are okay paying a dependence on the dimension, then there is no difference between the inverse-free and the inverse-ful models, because we can simulate an inverse query with $\bigTheta{d^2}$ forward queries~\cite{schl16,navascues18,sedlak2019optimal,qdssm19b,oym24,cmlzw24}.
Applying these techniques to the problems under consideration, we get that amplification and estimation can be done with $\bigO{d^2/a}$ and $\bigO{d^2/\eps}$ queries to $X$.
The same complexity can also be attained by learning the full unitary channel outright~\cite{hkot23}.
Somehow, though, these techniques do not beat the naive algorithms in the regime where $d$ is large, which is the typical regime.
If inverse access indeed provides an asymptotic improvement for these tasks, the manner in which it does so is quite subtle.

This mystery has practical relevance as well.
In many circumstances, access to $X$ is easy while access to $X^\dagger$ is hard.
Our standard rationale for being given access to both $X$ and $X^\dagger$ is as follows: we imagine that $X$ is given as a quantum circuit on a scalable quantum computer, in which case $X^\dagger$ can be performed by simply inverting the quantum circuit, gate by gate.
So, if applying $X$ is efficient, so is applying $X^\dagger$.
But amplification and estimation are also central tasks in the fields of quantum(-enhanced) sensing~\cite{glm04,drc17}, metrology~\cite{glm11,ta14}, and learning~\cite{slp11,aa23}, where $X$ arises from a natural process, like a device or a material in a lab.
In these domains, access to $X^\dagger$ can be much harder, morally as hard as reversing the direction of time within the system.
For example, the most prominent application of quantum-enhanced sensing is to the detection of gravitational waves~\cite{caves81,ligo13}: here, the unitary process is generated by, say, the merging of two black holes, so generating the inverse process seems impractical.

The primary goal in these domains, metrology especially, is to beat the naive $\bigO{1/\eps^2}$ algorithm for estimation.
With inverse access, essentially all such tasks achieve a Grover speedup.
But in these fields which cannot assume it, any kind of speedup seems hard to come by.
Are there better techniques in the inverse-less setting that we are missing?
Or do quadratic quantum speedups truly require inverse queries?

\subsection{Results}

We show that these quadratic speedups do require access to inverse queries.

\begin{theorem}[Amplitude amplification requires inverse access] \label{thm:lower-aa-main}
    Given only the ability to apply the input unitary $X \in \C^{d \times d}$, and not its inverse, any algorithm for amplitude amplification must use $\bigOmega{\min(d, 1/a^2)}$ applications of $X$.
\end{theorem}

\begin{theorem}[Amplitude estimation requires inverse access] \label{thm:lower-ae-main}
    Given only the ability to apply the input unitary $X \in \C^{d \times d}$, and not its inverse, any algorithm for amplitude estimation must use $\bigOmega{\min(d, 1/\eps^2)}$ applications of $X$.
\end{theorem}

Compare these lower bounds to the previously-discussed upper bounds of $\bigO{\min(d^2/a, 1/a^2)}$ and $\bigO{\min(d^2/\eps, 1/\eps^2)}$ for amplification and estimation, respectively.
When $d$ is large, as is typical, then the upper bounds match the lower bounds, showing that when we only have forward access to $X$, the naive algorithm---which uses non-adaptive measurements, no quantum space overhead, and no quantum gates beyond $X$---is query-optimal.
In many senses, without inverse access, the naive algorithm is optimal.

This confirms a conjecture of van Apeldoorn, Cornelissen, Gily\'{e}n, and Nannicini~\cite{vAcgn22}.
These lower bounds continue to hold even when given access to $\controlled X = \proj{0} \otimes I + \proj{1} \otimes X$ (\cref{rmk:ctrl}), by way of a procedure which converts controlled queries to their uncontrolled versions~\cite{tw25b}.
They also continue to hold even if we are told a constant-error estimate of $a$.
We refer the reader to \cref{rmk:what-is-it} for further discussion of how our results apply to variants of amplitude amplification and estimation.

Our lower bound explains the importance of inverse access in attaining generic quantum speedups for algorithmic tasks.
This result is particularly consequential in experimental domains, where access to the inverse of a process can be much more difficult than the process itself.
Recent work has noted that inverse access to a quantum process can make a problem tractable when it otherwise would not be~\cite{csm23,shh25}: they show this\footnote{
    These lower bounds are not unconditional: Cotler, Schuster, and Mohseni assume that the quantum algorithms have no ancilla qubits; Schuster, Haferkamp, and Huang assume the existence of a quantum-secure one-way function.
} for problems related to estimating out-of-time-ordered correlators (OTOCs), observables of circuits of the form $X A X^\dagger \ket{\psi}$.
Our work follows in a similar vein.
The central finding of the fields of quantum-enhanced sensing and metrology is that estimating certain parameters of a quantum-mechanical system can be done quadratically more efficiently, if we are allowed to use entanglement in our estimation procedures.
These procedures only need forward queries.
We give an explicit estimation task where quadratic savings is only possible when we have \emph{both} forward and inverse queries.
Prior no-go theorems only rule out such quadratic savings in the presence of noise or non-unitarity~\cite{jwdfy08,glm11}.
We show that there is a more fundamental limitation beyond noise: if our experiments must follow the arrow of time, then the techniques of quantum sensing and metrology will always remain limited in scope, despite the generic, sweeping speedup for estimation tasks that we see in the field of quantum algorithms.

Our proofs of \cref{thm:lower-aa-main,thm:lower-ae-main} are simple and clearly identify the source of the quadratic separation between forward and inverse queries.
We found this surprising, as the technical barriers for this proof seemed difficult to surmount.

\paragraph{Proving lower bounds against forward queries.}
The central barrier of this work is the lack of existing methods which are suitable for attacking this lower bound.
The authors who conjectured the lower bound remark on this barrier, stating that they ``are not aware of any lower bound techniques that differentiate between normal and inverse usage of an input oracle''~\cite{vAcgn22}.

The overwhelming majority of lower bound techniques in quantum query complexity have been developed in the setting where one is given access to an oracle for a string $x \in \mathcal{D}^n$.
When the domain $\mathcal{D}$ is Boolean, this oracle will be its own inverse,
and when the domain is larger, there are multiple ways to define this oracle,
all of which lead to unitary matrices which either are their own inverse or can easily simulate their own inverse.
As a result, these lower bound techniques cannot tell the difference between an $X$ query and an $X^{\dagger}$ query, at least in their native settings.
Recently, the polynomial method~\cite{bbcmw01} was generalized to hold for arbitrary unitary queries by She and Yuen~\cite{sy23},
but in defining this generalization they specifically only considered algorithms which are granted access to both $X$ and $X^\dagger$.
More promisingly, the adversary method~\cite{ambainis00} has also been generalized to hold for arbitrary unitary queries by Belovs~\cite{belovs15},
and his generalization allows for algorithms which are not given inverse access.
Indeed, as we will describe below, this generalized adversary method \emph{is} actually strong enough to capture the difference between a query to $X$ and a query to its inverse.
However, the adversary method is famously unwieldy to use, and we do not have many tools for analyzing this particular generalization of it.

To our knowledge, there are only a few problems where there is a known separation between the inverse-less and inverse-ful settings.
The first of these is
permutation inversion with an in-place oracle,
stated as follows: given a unitary $P \in \C^{d \times d}$ which implements a permutation, $P\ket{i} = \ket{\pi(i)}$, find $\pi^{-1}(1)$.
This can be solved with one query to $P^\dagger$, but Fefferman and Kimmel~\cite{fk18} show, using a variant of the adversary method adapted to permutation matrices, that with only access to $P$, $\bigOmega{\sqrt{d}}$ forward queries are necessary.
This was later reproved by Belovs and Yolcu~\cite{by23} using the adversary bound of Belovs~\cite{belovs15}.
Together, these results show that suitable generalizations of the adversary method \emph{can} distinguish between the inverse-free and inverse-ful access models.
Fefferman and Kimmel then conjectured that $\bigOmega{d}$ forward queries were necessary to solve permutation inversion, meaning that quantum computers could not beat classical brute force.
This conjecture's logic follows the same logic as ours: uncomputation enables quadratic speedups, so without the uncomputation there is no quadratic speedup.
However, Holman, Ramachandran, and Yirka~\cite{hry25} recently disproved Fefferman and Kimmel's conjecture by giving a $\bigO{\sqrt{d}}$-query algorithm for permutation inversion.
This work provides an interesting counter-point to our work: although permutation inversion is another problem where inverse access can be useful,
their work shows that a ``Grover-like'' quadratic speedup is still possible even without inverse access.

For amplitude amplification and estimation, we separate the inverse-less and inverse-ful models via the compressed oracle method of Zhandry~\cite{zhandry19}.
We were inspired to try this method upon seeing its success in proving the existence of pseudorandom unitaries (PRUs).
An efficiently computable random unitary $X$ is a PRU if it is indistinguishable from a Haar random unitary to any efficient adversary given query access to $X$,
and it is a strong PRU if it remains secure when the adversary is also given query access to $X^{\dagger}$.
Ma and Huang~\cite{mh25} gave the first construction of PRUs and strong PRUs with a proof based on the compressed oracle method.
A striking feature of this method is its simplicity and clarity, particularly when compared to prior work whose primary tool was representation theory~\cite{mpsy24,shh25}.
This method further appealed to us because it had already been demonstrated to be sensitive to the difference between forward and inverse queries.
Ma and Huang gave examples of PRUs which are provably not strong PRUs, so the task of breaking these PRUs therefore gives a (conditional) separation between the inverse-less and inverse-ful models.
Our work further illustrates the power of this method: not only is it able to (unconditionally) separate these two models, the resulting proof is technically clean and conceptual.

\paragraph{Proof overview.}
To describe how we use the compressed oracle method, we will first sketch the high-level structure of the proof.
We will focus on the lower bound for amplitude estimation, from which the lower bound for amplitude amplification follows as a corollary.

We prove hardness for the problem of (normalized) \emph{trace estimation}: given a unitary $U \in \C^{d \times d}$, estimate $\abs{\ntr(U)} = \frac{1}{d} \abs{\tr(U)}$ to $\eps$ error.
Trace estimation reduces to amplitude estimation, while inheriting its pertinent features.
With only access to $U$, the best-known algorithm for trace estimation queries $U$ $\bigO{\min(d^2/\eps,1/\eps^2)}$ times, where the two parts of the complexity come from unitary channel tomography~\cite{hkot23} and amplitude estimation~\cite{kl98}, respectively.
With $U$ and $U^\dagger$, we can use the subroutine for amplitude estimation to get an algorithm with $\bigO{1/\eps}$ queries to $U$ and $U^\dagger$.
We further expect this task to be hard without access to $U^\dagger$, because the trace is a global feature of the unitary, and so might be particularly hard to estimate without inverses.

Specifically, we define two ensembles of random diagonal unitaries (\cref{def:ae-ensembles}).
Ensemble 0 is unbiased, meaning that its entries are independent phases with mean zero.
Ensemble 1 is biased, so its entries are independent phases biased towards $1$.
We consider the distinguishing task (\cref{prob:ae-distinguish}): given access to $U^{(b)}$ drawn from Ensemble $b \in \braces{0,1}$, decide which $b$.
\begin{align} \label{eq:intro-ensembles}
        U^{(0)} &= \begin{bmatrix}
        U_1^{(0)} \\ & \ddots \\ & & U_d^{(0)}
    \end{bmatrix} \text{ with } \E[U_i^{(0)}] = 0 \quad \text{\textbf{or}} \quad
        U^{(1)} = \begin{bmatrix}
        U_1^{(1)} \\ & \ddots \\ & & U_d^{(1)}
    \end{bmatrix} \text{ with } \E[U_i^{(1)}] \approx \eps
\end{align}
This task can be solved with trace estimation, since with good probability, $\abs{\ntr(U^{(0)})} < 0.1\eps$, while $\abs{\ntr(U^{(1)})} \geq 0.2\eps$ (\cref{lem:traces-distinguish,prop:ae-reduction}).

We show a lower bound of $\bigOmega{1/\eps^2}$ queries to $U$ for the distinguishing task, following the compressed oracle method (\cref{thm:lower-ae}).\footnote{
    The lower bound is $\bigOmega{\min(d, 1/\eps^2)}$ because the two ensembles are only distinguishable provided $d \gg 1/\eps^2$: only then does the normalized trace concentrate.
}
We consider a quantum circuit which applies $U$ $n$ times and attempts to distinguish the ensembles.
This algorithm outputs a particular quantum state $\ket{\alg_U}$ on the input unitary $U$.
So, if Ensemble $b$ samples a unitary $U$ with probability $p_b(U)$, then the final state of the algorithm when run on this random $U$ is
\begin{align*}
    \tr_\reg{P}\parens[\Big]{\proj{\alg^{(b)}}} \text{ where } \ket{\alg^{(b)}} = \sum_U \sqrt{p^{(b)}(U)} \ket{\alg_U}\ket{U}_{\reg{P}}.
\end{align*}
The compressed oracle method gives us a way to identify the components of this mixed state: they are sums of Feynman paths through the circuit.
Further, we can couple the components between the two ensembles.
Although $\ket{\alg^{(0)}}$ and $\ket{\alg^{(1)}}$ are far from each other, we find unitaries $Q^{(0)}$ and $Q^{(1)}$ such that
\begin{align*}
    Q^{(0)}_{\reg{P}} \ket{\alg^{(0)}} \approx Q^{(1)}_{\reg{P}}\ket{\alg^{(1)}}.
\end{align*}
These unitaries are biased Fourier transforms (\cref{lem:biased-ft}).
Even though we apply different unitaries to these two states, they are on the purification register, so the final mixed states attained by tracing out this register remain the same.
This coupling reveals a special orthogonality behavior in the purifications that we exploit in our proof (\cref{lem:eps-squared}).
In the language of mixed states, it turns out that the difference between the mixed states of each ensemble can be captured in the \emph{probabilities} of the individual components in the mixture.
This is only true for algorithms with forward access to $U$; when $U^\dagger$ is given as well, the same analysis can be followed, but the mixed states will differ, not just in the probabilities of the components, but in the components themselves.
The deviation being ``classical'' under forward access is what gives the lower bound of $1/\eps^2$: applying only $U$ without $U^\dagger$ creates de-coherence within the quantum algorithm which dashes our hopes for a quadratic speedup.

\subsection{Discussion} \label{subsec:discussion}

We have shown that the quadratic speedup of amplitude amplification and estimation cannot be achieved without access to both $X$ and $X^\dagger$.
Specifically, we show a lower bound on amplitude estimation (and normalized trace estimation) of $\bigOmega{\min(d, 1/\eps^2)}$, compared to the best known upper bound of $\bigO{\min(d^2/\eps, 1/\eps^2)}$.
We conjecture that the upper bound is tight, suggesting that there are two regimes in the inverse-less setting.
\emph{Estimating all features} of a unitary channel $X$, i.e.\ learning the unitary, can be done coherently, with a $\bigO{1/\eps}$ scaling.
For \emph{estimating one desired feature} of $X$, though, nothing is better than the naive, incoherent $\bigO{1/\eps^2}$ algorithm.
Improving our lower bound may require changing the distinguishing task.

This work fits into a broader line of research giving barriers to Grover's algorithm.
For example, we know that a quadratic Grover speedup is not possible if the query oracle has a chance of erring~\cite{rs08} or if we only have access to low-depth circuits~\cite{Zalka1997}.
This work demonstrates a natural setting where some, but not all, Grover speedups can be performed.
For $X$ a diagonal unitary matrix, we can estimate \emph{local} information about $X$, like a relative phase $X_i / X_1$, to $\eps$ error with $\bigO{1/\eps}$ forward queries.
But our work shows that estimating \emph{global} information requires $\bigOmega{1/\eps^2}$ forward queries, if we are not willing to pay a dimension dependence.
This tracks with our intuitive assessment from the introduction: algorithms with forward queries will wander out into $d$-dimensional Hilbert space, but if we are allowed $\poly(d)$ queries, we can handle this entire space, allowing us to achieve our desired $\bigO{1/\eps}$ scaling.

Looking forward, we see our results as relevant for developing an understanding of unitary complexity.
As the field of quantum algorithms has matured, its central ideas have been most effectively expressed as ``unitary problems'', problems where the input is black-box access to a unitary $U$.
This has led to a shift towards building notions of complexity where the unitary is a first-class object~\cite{lmw24,bempqy23,zhandry25}.
However, a basic question one encounters when defining these notions is: should we by default give access to just $U$, or $U$ and $U^\dagger$?
On one hand, for many problems inspired by experiment, we would like to explicitly forbid inverse queries when considering hardness.
On the other hand, we prove that, without inverse queries, we lose the most generic form of quantum speedup, which such a theory ought to capture.

\paragraph{Notation.}
We denote $\ii = \sqrt{-1}$.
We use sans serif script to refer to registers of a quantum state: for example, $\ket{\psi}_{\reg{P}}\ket{\phi}_{\reg{RS}}$ denotes the system where $\ket{\psi}$ occupies register $\reg{P}$ and $\ket{\phi}$ occupies registers $\reg{R}$ and $\reg{S}$.
We drop the register notation when the set of registers is clear from context.

\section{Defining the hard instance}

\begin{remark}[What about other versions of amplitude amplification and estimation?] \label{rmk:what-is-it}
    According to our definitions~(\cref{def:aa-ae}), in amplitude amplification and estimation, we are not given any information about $a$.
    We do this for simplicity: all our proofs continue to hold if we are also given an estimate $a^* > 0$ with the guarantee that $a^* \leq a \leq Ca^*$ for some large constant $C$.
    In fact, our $\bigOmega{1/\eps^2}$ lower bound for amplitude estimation should hold for any constant $a^*$, by padding our hard instance (\cref{def:ae-ensembles}) with $\phi I$ for $\phi \sim \mu_0$ and extending the proofs accordingly.
    Amplitude amplification is even a well-defined problem when we receive $a$ exactly, but because of technical limitations, our lower bounds do not extend to this setting.

    Our definition of amplitude estimation is estimation to \emph{additive} error.
    Amplitude estimation algorithms are typically stated with a relative-error guarantee, where the algorithm uses $\bigO{1/(a\eps)}$ queries to $U$ and $U^\dagger$ (or $\bigO{1/(a \eps)^2}$ queries to $U$ alone) to get an estimate $(1-\eps)a < \wh{a} < (1+\eps)a$.
    However, this guarantee is identical to the additive-error one we state, as can be seen by rescaling $a\eps \gets \eps$, so our lower bound applies to both.\footnote{
        It seems that relative-error is the prevailing version of the guarantee only because classical intuition suggests it could be stronger: estimating the bias of a coin with probability $p$ of heads to $\eps$ relative error requires $\bigO{\frac{1}{p\eps^2}}$ samples, so one might expect the dependence on $a$ to differ from the dependence on $\eps$.
        This turns out to not be the case.
    }
    We use the additive-error version because it removes the $a$ dependence from the complexity which distracts from the fundamental quadratic speedup in $\eps$.
\end{remark}

Throughout, we let $q$ be a positive integer, and we let $\omega = e^{2\pi \ii / q}$ be an order-$q$ root of unity.
Then the set $\braces{\omega^k}_{k \in [q]}$ forms a group under multiplication, which we call $\cyc_q$.
All that we will need of $q$ is that it is sufficiently large; for our results, we could take $q \to \infty$ without issue.
We work with finite $q$ for simplicity.

\begin{definition}[Biased and unbiased distributions over roots of unity]
    The \text{$\eps$-biased} distribution over phases, $\mu_\eps$, is defined in the following way.
    Let $\mu_0: \cyc_q \to \mathbb{R}$ be the uniform measure over $\cyc_q$: $\mu_0(g) = 1/q$ for all $g \in \cyc_q$.
    Let $\mu_1: \cyc_q \to \mathbb{R}$ be the uniform measure over the elements of $\cyc_q$ on the right half of the unit circle, $\braces{\omega^{-\floor{q/4}}, \dots, \omega^{\floor{q/4}}}$.
    Then $\mu_\eps$ is the linear interpolation between the two,
    \begin{align*}
        \mu_\eps(g) = (1 - \eps) \mu_0(g) + \eps \mu_1(g)
        = \begin{cases}
            \frac{1}{q}(1 + \eps\parens[\big]{\frac{q}{2\floor{q/4}+1} - 1 }) & g \in \braces{\omega^{-\floor{q/4}}, \dots, \omega^{\floor{q/4}}} \\
            \frac{1}{q}(1 - \eps) & \text{otherwise}
        \end{cases}
    \end{align*}
\end{definition}

\begin{lemma}[Expectations of the distributions] \label{lem:circle-means}
    For sufficiently large $q$, we have that
\begin{gather*}
    \E_{g \sim \mu_0} [g] = 0 \quad\text{and}\quad
    \frac{1}{2} < \E_{g \sim \mu_1} [g] < 1.
\end{gather*}
\end{lemma}
\begin{proof}
For the first equation,
\begin{align*}
    \E_{g \sim \mu_0} [g] = \frac{1}{q} \sum_{g \in \cyc_q} g = 0.
\end{align*}
As for $\mu_1$, we first observe that the expectation is always real; then, we bound it.
\begin{align*}
    \E_{g \sim \mu_1} [g]
    &= \frac{1}{2\floor{q/4} + 1} \sum_{k = -\floor{q/4}}^{\floor{q/4}} \omega^{k} \\
    &= \frac{1}{2\floor{q/4} + 1} \parens[\big]{1 + 2\cos(2\pi/q) + 2\cos(4\pi/q) + \dots + 2\cos(2\pi \floor{q/4}/q)} \\
    &= \frac{1}{2\floor{q/4} + 1} \frac{\sin((2\floor{q/4} + 1) \pi/q)}{\sin(\pi/q)}
\end{align*}
The last step follows from a standard derivation: the sum appearing above is known as the Dirichlet kernel.
So, we can take a limit:
\begin{align*}
    \lim_{q \to \infty} \E_{g \sim \mu_1} [g]
    = \lim_{q \to \infty}\frac{1}{(2\floor{q/4} + 1)\sin(\pi/q)}
    = \frac{2}{\pi}.
\end{align*}
Thus, for sufficiently large $q$, $1/2 < \E_{g \sim \mu_1}[g] < 1$ as desired.
\end{proof}

\begin{definition}[Diagonal unitary ensembles] \label{def:ae-ensembles}
    For a dimension $d$, a bias parameter $\eps \in (0,1)$, and an order parameter $q$, we consider two different ensembles of unitary matrices.
    In both ensembles, $U \in \mathbb{C}^{d \times d}$ is a random diagonal matrix where
    \begin{enumerate}
        \setcounter{enumi}{-1}
        \item $U$'s diagonal entries, $U_i$, are drawn i.i.d.\ from $\mu_0$;
        \item $U$'s diagonal entries, $U_i$, are drawn i.i.d.\ from $\mu_\eps$.
    \end{enumerate}
    We call these ensemble 0 (the unbiased ensemble) and ensemble 1 (the biased ensemble).
\end{definition}

We consider the problem of distinguishing these two ensembles.

\begin{problem}[Distinguishing unitary ensembles] \label{prob:ae-distinguish}
    For parameters $d, \eps, q$, consider the following problem.
    Let $b \in \braces{0,1}$ be a bit which is zero or one with probability $1/2$.
    Then draw a unitary $U$ from ensemble $b$ (\cref{def:ae-ensembles}).
    The goal is to, given quantum access to the unitary $U$, output a bit $\wh{b} \in \braces{0,1}$ such that $\wh{b} = b$ with good probability.
\end{problem}

This problem is inverse-invariant: for $U$ drawn from one of the ensembles, $U^\dagger$ follows the same ensemble.
Therefore, the problem is exactly as difficult with only forward queries as it is with only inverse queries, in contrast to problems like permutation inversion (see \cref{subsec:discussion}).

We begin by showing that this problem can be solved using amplitude estimation, and therefore enjoys a quadratic quantum speedup when we have access to $U$ and $U^\dagger$.

\begin{lemma}[Normalized traces concentrate] \label{lem:traces-concentrate}
    Let $V \in \C^{d \times d}$ be a random diagonal unitary matrix, whose diagonal entries $V_i$ are drawn from independent distributions.
    Then, provided $d \geq C / t^2$ for $C$ a sufficiently large constant,
    \begin{align*}
        \Pr\bracks[\Big]{\abs[\Big]{\ntr(V) - \E[\ntr(V)]} \geq t} \leq 0.01.
    \end{align*}
\end{lemma}
\begin{proof}
This is a consequence of Hoeffding's inequality~\cite[Theorem 2.2.6]{Ver18}, which states that for independent random variables $X_1,\dots,X_d$ such that $X_i \in [-1,1]$ for every $i$, then
\begin{align*}
    \Pr\bracks[\Big]{\abs[\Big]{\frac{1}{d}\sum_{i=1}^d (X_i - \E[X_i])} \geq t} \leq 2e^{-dt^2/2}.
\end{align*}
Then $\Re(V_i)$ are independent random variables between $-1$ and $1$, and similarly for $\Im(V_i)$.
So,
\begin{align*}
    &\Pr\bracks[\Big]{\abs[\Big]{\ntr(V) - \E[\ntr(V)]} \geq t} \\
    &\leq \Pr\bracks[\Big]{\abs[\Big]{\Re(\ntr(V) - \E[\ntr(V)])} \geq t/2} + \Pr\bracks[\Big]{\abs[\Big]{\Im(\ntr(V) - \E[\ntr(V)])} \geq t/2} \\
    &\leq 2e^{-d(t/2)^2/2} + 2e^{-d(t/2)^2/2} \leq 0.01.
\end{align*}
In the final step, we used that $d \geq C/t^2$ for $C$ sufficiently large.
\end{proof}

\begin{lemma}[Normalized trace is a distinguisher for the ensembles] \label{lem:traces-distinguish}
    Let $U^{(0)}$ and $U^{(1)}$ be random unitary matrices drawn from ensemble 0 and ensemble 1 of \cref{def:ae-ensembles}, respectively, where $\eps > 0$, $d \geq C / \eps^2$, and $q \geq C$ for $C$ a sufficiently large constant.
    Then
    \begin{enumerate}
        \setcounter{enumi}{-1}
        \item With probability $\geq 0.99$, $\abs{\ntr(U^{(0)})} < 0.1\eps$;
        \item With probability $\geq 0.99$, $\abs{\ntr(U^{(1)})} \geq 0.2\eps$.
    \end{enumerate}
\end{lemma}
\begin{proof}
Recall from \cref{def:ae-ensembles} that $U^{(0)}$ and $U^{(1)}$ are diagonal unitary matrices whose entries are independent random variables.
Further, we have that $d \geq C/(\eps^2)$ for $C$ sufficiently large, so we can apply \cref{lem:traces-concentrate} to get that
\begin{align}
    \Pr\bracks[\Big]{\abs[\Big]{\ntr(U^{(0)}) - \E[\ntr(U^{(0)})]} \geq 0.1 \eps} &\leq 0.01;
    \label{eq:tr0-concentrates} \\
    \Pr\bracks[\Big]{\abs[\Big]{\ntr(U^{(1)}) - \E[\ntr(U^{(1)})]} \geq 0.1 \eps} &\leq 0.01.
    \label{eq:tr1-concentrates}
\end{align}
It remains to compute these expectations.
By \cref{lem:circle-means},
\begin{align*}
    \E[\ntr(U^{(0)})] &= \E_{g \sim \mu_0} [g] = 0 \\
    \E[\ntr(U^{(1)})]
    &=\E_{g \sim \mu_\eps} [g]
    = (1 - \eps)\E_{g \sim \mu_0}[g] + \eps\E_{g \sim \mu_1}[g]
    \in (\eps/2, \eps).
\end{align*}
This, together with \eqref{eq:tr0-concentrates}, implies that, with probability $\geq 0.99$,
\begin{align*}
    \abs{\ntr(U^{(0)})}
    &= \abs{\ntr(U^{(0)}) - \E[\ntr(U^{(0)})]}
    < 0.1 \eps.
\end{align*}
Next, we consider the behavior of $\ntr(U^{(1)})$.
Because of \eqref{eq:tr1-concentrates}, with probability $\geq 0.99$, we have that
\begin{equation*}
    \abs{\ntr(U^{(1)})}
    \geq \abs{\E[\tr(U^{(1)})]} - \abs{\tr(U^{(1)}) - \E[\tr(U^{(1)})]}
    \geq 0.5 \eps - 0.1 \eps \geq 0.2\eps. \qedhere
\end{equation*}
\end{proof}

\begin{proposition} \label{prop:ae-reduction}
    Consider the distinguishing task, \cref{prob:ae-distinguish}, with parameters $\eps > 0$, $d \geq C/\eps^2$, and $q \geq C$ with $C$ a sufficiently large constant.
    Then there is an algorithm which succeeds at this task with probability $\geq 0.9$, using one call to amplitude estimation with error parameter $0.05\eps$ (\cref{def:aa-ae}).
\end{proposition}
\begin{proof}
Let $T \in \C^{d \times d}$ be a unitary matrix which prepares the uniform superposition, $T \ket{0} = \frac{1}{\sqrt{d}} \sum_{i=0}^{d-1} \ket{i}$, and let $Z \in \C^{2d \times 2d}$ be a unitary matrix which ``marks'' $\ket{0}$: $Z \ket{0} \ket{j} = \ket{0}\ket{j \oplus 1}$, and $Z \ket{i}\ket{j} = \ket{i}\ket{j}$ for $i \neq 0$.
Then the quantum circuit
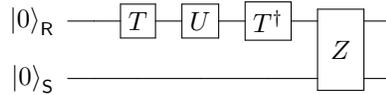
\begin{align*}
    \Qcircuit @R=1em @C=1em {
        \lstick{\ket{0}_{\reg{R}}} & \qw & \gate{T} & \gate{U} & \gate{T^\dagger} & \multigate{1}{Z} & \qw \\
    \lstick{\ket{0}_{\reg{S}}} & \qw & \qw & \qw & \qw & \ghost{Z} & \qw
    }
\end{align*}
outputs the state
\begin{align*}
    \ntr(U) \ket{0}_{\reg{R}} \ket{1}_{\reg{S}} + \sqrt{1 - \abs{\ntr(U)}^2} \ket{\phi}_{\reg{R}} \ket{0}_{\reg{S}}.
\end{align*}
By calling amplitude estimation (\cref{def:aa-ae}) on this circuit, $Z_{\reg{RS}} T^\dagger_{\reg{R}} U_{\reg{R}} T_{\reg{R}}$, we get an estimate $\wh{a}$ such that, with probability $\geq 0.99$, $\abs{\ntr(U)} - 0.05\eps < \wh{a} < \abs{\ntr(U)} + 0.05\eps$.

By \cref{lem:traces-distinguish}, when $U$ comes from ensemble 0, $\wh{a} \in (-0.15\eps, 0.15\eps)$ with probability $\geq 0.98$.
On the other hand, when $V$ comes from ensemble 1, $\wh{a} \not\in (-0.15\eps, 0.15\eps)$ with probability $\geq 0.98$.
So, if we output $\wh{b} = 0$ when $\wh{a} \in (-0.15\eps, 0.15\eps)$, and $\wh{b} = 1$ otherwise, then $\wh{b}$ is the correct ensemble with probability $\geq 0.9$.
\end{proof}

\section{Proving the lower bound for amplitude estimation}

In this section, we prove the following theorem.

\begin{theorem} \label{thm:lower-ae}
    Consider \cref{prob:ae-distinguish} with $\eps \in (0, 1/2)$ and $q \geq \frac{1}{\eps^2}$.
    Then any algorithm to solve it with probability $\geq \frac23$ requires $n \geq \frac{1}{12\eps^2}$ applications of $U$.
\end{theorem}

With this, we can prove our main theorem as a corollary.

\begin{proof}[Proof of \cref{thm:lower-ae-main}]
Let $C$ be the sufficiently large constant from \cref{prop:ae-reduction}.
We will assume that $\epsilon \in (0,0.025)$ and that $d > 4C$, as the theorem is vacuously true when these bounds do not hold.
Let
\begin{equation*}
    \epsilon^+ = \max\{\epsilon/0.05, \sqrt{C/d}\},
\end{equation*}
and note that $\epsilon^+ \in (0, 1/2)$ by our assumptions on $\epsilon$ and $d$.
Consider the distinguishing task from \cref{prob:ae-distinguish} with bias parameter $\epsilon^+$
and order parameter
$q = \lceil\max\{1/\eps^2, C\}\rceil$.
As $d \geq C/(\epsilon^+)^2$ and $q \geq C$, we can apply \cref{prop:ae-reduction},
which states that there is an algorithm which succeeds at this distinguishing task with probability $\geq 0.9$ using one call to amplitude estimation.
For this algorithm, it suffices to call amplitude estimation with error parameter $\epsilon$, as $\epsilon \leq 0.05 \epsilon^+$.
But since $\epsilon^+ \in (0, 1/2)$ and $q \geq 1/\epsilon^2$, we can apply \cref{thm:lower-ae}, which states that this algorithm requires
\begin{equation*}
    n \geq \frac{1}{12(\epsilon^+)^2}
    = \Omega(\min(d,1/\eps^2))
\end{equation*}
applications of $U$.
As the only applications of $U$ in the algorithm occur in the amplitude estimation subroutine,
amplitude estimation with error $\epsilon$ requires this many applications of $U$.
\end{proof}

The crux of this proof is to extract a $\bigOmega{1/\eps^2}$ lower bound from a purification of the algorithm, when run on the two different ensembles of \cref{prob:ae-distinguish}.
When the ensemble changes, the only thing that changes in the output state is that the input unitary $U$ is sampled from a different distribution; so, the difference in the output states is located solely in the purification register.
What matters is whether this difference is ``coherent'' or ``incoherent'': does the difference affect the states in the final distribution, or only the probabilities of these states?
For example, consider the un-normalized state
\begin{align*}
    \ket{\psi}_{\reg{OP}} = \ket{0}_{\reg{O}}\ket{0}_{\reg{P}} + \ket{1}_{\reg{O}}\ket{1}_{\reg{P}}.
\end{align*}
We can perturb the purification register by $\eps$, to get
\begin{align*}
    \ket{\psi'}_{\reg{OP}} = \ket{0}_{\reg{O}}\parens[\big]{\sqrt{1-\eps^2}\ket{0} + \eps\ket{\bot}}_{\reg{P}} + \ket{1}_{\reg{O}}\parens[\big]{\sqrt{1-\eps^2}\ket{1} + \eps\ket{\bot}}_{\reg{P}}.
\end{align*}
Then one can verify that $\tr_\reg{P}(\ketbra{\psi}{\psi}) = \bracks[\big]{\begin{smallmatrix} 1 & 0 \\ 0 & 1 \end{smallmatrix}}$ and $\tr_\reg{P}(\ketbra{\psi'}{\psi'}) = \bracks[\big]{\begin{smallmatrix} 1 & \eps^2 \\ \eps^2 & 1 \end{smallmatrix}}$ differ by $\bigTheta{\eps^2}$ in trace distance.
On the other hand, if we change the purification register by $\eps$ in a different fashion,
\begin{align*}
    \ket{\psi''}_{\reg{OP}} = \ket{0}_{\reg{O}}\parens[\big]{\sqrt{1-\eps^2}\ket{0} + \eps\ket{1}}_{\reg{P}} + \ket{1}_{\reg{O}}\parens[\big]{\sqrt{1-\eps^2}\ket{1} + \eps\ket{0}}_{\reg{P}}.
\end{align*}
Then $\tr_\reg{P}(\ketbra{\psi''}{\psi''}) = \bracks[\big]{\begin{smallmatrix} 1 & 2\eps\sqrt{1-\eps^2} \\ 2\eps\sqrt{1-\eps^2} & 1 \end{smallmatrix}}$ differs from $\tr_\reg{P}(\ketbra{\psi}{\psi})$ by $\bigTheta{\eps}$ in trace distance.
In other words, if the $\eps$-sized deviation between the two purifications arises in a part of the space which is orthogonal to the purification itself, then this only affects the state after tracing out by $\bigO{\eps^2}$, which is quadratically smaller than typical.
We capture this phenomenon in the following lemma.

\begin{lemma}[Controlling orthogonal errors in the purification register] \label{lem:eps-squared}
    Let $\ket{\psi}_{\reg{OP}} = \sum_{i=1}^r \ket{\phi_i}_{\reg{O}} \ket{\varphi_i}_\reg{P}$ be a quantum state (i.e.\ $\braket{\psi}{\psi} = 1$), where the $\braces{\ket{\varphi_i}}_{i}$'s are orthonormal (and the $\ket{\phi_i}$'s can have norm smaller than one).
    Let $\ket{\wt{\psi}}_{\reg{OP}}$ be a state which differs from $\ket{\psi}_{\reg{OP}}$ only on register $\reg{P}$:
    \begin{align*}
        \ket{\wt{\psi}}_{\reg{OP}} = \sum_{i=1}^r \ket{\phi_i}_\reg{O} (\alpha_i\ket{\varphi_i}_\reg{P} + \ket{\err_i}_\reg{P}).
    \end{align*}
    Further suppose that, for all $i$, $\abs{\alpha_i}^2 \geq 1 - \delta$, and for all $i$ and $j$, $\braket{\varphi_i}{\err_j} = 0$.
    Then
    \begin{align*}
        \frac12 \trnorm[\Big]{\tr_\reg{P}(\proj{\psi} - \proj{\wt{\psi}})} \leq \delta.
    \end{align*}
\end{lemma}
\begin{proof}
Consider measuring the purification register of $\ket{\psi}_{\reg{OP}}$ and $\ket{\wt{\psi}}_{\reg{OP}}$ with the POVM $\braces{O_1, \dots, O_r, \ol{O}}$, where $O_i = \proj{\varphi_i}$ for $i \in [r]$ and $\ol{O} = I - \sum_i \proj{\varphi_i}$, and then discarding the measurement outcome.
The resulting state is equal to the state we would get by directly tracing out register $\reg{P}$.
In other words, for any state $\rho_{\reg{OP}}$,
\begin{align*}
    \tr_\reg{P}[\rho_{\reg{OP}}] = \tr_{\reg{P}}\bracks[\big]{(\ol{O})_{\reg{P}} \rho_{\reg{OP}}} + \sum_{i=1}^r \tr_{\reg{P}}\bracks[\big]{(O_i)_{\reg{P}} \rho_{\reg{OP}}}.
\end{align*}
The probability we see the outcome $O_i$ from $\ket{\psi}_{\reg{OP}}$ is $\braket{\phi_i}{\phi_i}$, in which case the resulting state is $\proj{\phi_i} / \braket{\phi_i}{\phi_i}$.
Moreover, these probabilities sum to one, $\sum_i \braket{\phi_i}{\phi_i} = 1$, so the POVM element $\ol{O}$ never appears.

As for $\ket{\wt{\psi}}_{\reg{OP}}$, the probability we see the outcome $\proj{\varphi_i}$ upon measuring is $\abs{\alpha_i}^2\braket{\phi_i}{\phi_i}$, and the resulting state is identical to the outcome from $\ket{\psi}$, $\proj{\phi_i} / \braket{\phi_i}{\phi_i}$.
The POVM element $\ol{O}$ appears with the remaining probability,
\begin{align*}
    1 - \sum_{i=1}^r \abs{\alpha_i}^2 \braket{\phi_i}{\phi_i}
    \leq 1 - \sum_{i=1}^r (1 - \delta) \braket{\phi_i}{\phi_i}
    = \delta.
\end{align*}
In summary, we have that
\begin{align*}
    & \trnorm[\Big]{\tr_\reg{P}(\proj{\psi}) - \tr_{\reg{P}}(\proj{\wt{\psi}})} \\
    &\leq \trnorm[\Big]{\tr_\reg{P}(\ol{O}_\reg{P}\proj{\psi}) - \tr_{\reg{P}}(\ol{O}_{\reg{P}}\proj{\wt{\psi}})}
    + \sum_{i=1}^r \trnorm[\Big]{\tr_\reg{P}((O_{i})_{\reg{P}}\proj{\psi}) - \tr_{\reg{P}}((O_{i})_{\reg{P}}\proj{\wt{\psi}})} \\
    &= \trnorm[\Big]{\tr_{\reg{P}}(\ol{O}_{\reg{P}}\proj{\wt{\psi}})}
    + \sum_{i=1}^r \trnorm[\Big]{\proj{\phi_i} - \abs{\alpha_i}^2\proj{\phi_i}} \\
    &\leq \delta + \sum_{i=1}^r (1 - \abs{\alpha_i}^2)\braket{\phi_i}{\phi_i} \\
    &\leq 2\delta. \qedhere
\end{align*}
\end{proof}

\begin{lemma}[$\eps$-biased Fourier transform] \label{lem:biased-ft}
    For $\eps \in (0, 1/2)$, there is a unitary matrix $F^{(\eps)} \in \mathbb{C}^{q \times q}$ with the following property.
    For all $k \in \braces{0,\dots,q-1}$,
    \begin{align*}
        F^{(\eps)}\parens[\Big]{\sum_{g \in \cyc_q} \sqrt{\mu_\eps(g)} g^k \ket{g}}
        = \alpha_k \ket{k} + \sum_{\ell = 0}^{k - 1} \alpha_\ell \ket{\ell}.
    \end{align*}
    Here, $\alpha_0,\dots,\alpha_k \in \C$ have the additional property that $\sum_{\ell=0}^k \abs{\alpha_\ell}^2 = 1$ and $\abs{\alpha_k}^2 \geq 1 - 4\eps^2$.
\end{lemma}

\begin{proof}[Proof of \cref{thm:lower-ae}]
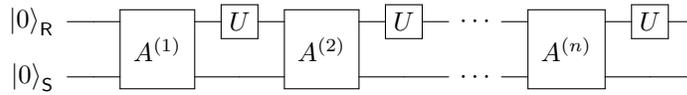
\begin{figure}[ht]
\[ 
\Qcircuit @R=1em @C=1em {
\lstick{\ket{0}_{\reg{R}}}
    & \qw
    & \multigate{1}{A^{(1)}}
    & \gate{U}
    & \multigate{1}{A^{(2)}}
    & \gate{U}
    & \qw
    & {\cdots}
    &
    & \multigate{1}{A^{(n)}}
    & \gate{U}
    & \qw  \\
\lstick{\ket{0}_{\reg{S}}}
    & \qw
    & \ghost{A^{(1)}}
    & \qw
    & \ghost{A^{(2)}}
    & \qw
    & \qw
    & {\cdots}
    &
    & \ghost{A^{(n)}}
    & \qw
    & \qw  
}
\]
\caption{
    The form of an arbitrary distinguishing algorithm which queries $U$ $n$ times to solve \cref{prob:ae-distinguish}.
}
\label{fig:alg}
\end{figure}

We consider an arbitrary distinguishing algorithm which applies the given unitary $U$ $n$ times.
Without loss of generality, we can write such a computation in the following way: prepare the pure state according to
\begin{align*}
    \ket{\alg_U}_\reg{RS} \coloneqq U_\reg{R} A^{(n)}_{\reg{RS}} \dots U_{\reg{R}} A^{(2)}_{\reg{RS}} U_{\reg{R}} A^{(1)}_{\reg{RS}} \ket{0}_{\reg{RS}},
\end{align*}
and then perform a binary measurement; if the output is 0, conclude that $U$ came from ensemble $0$, and when the output is 1, conclude that $U$ came from ensemble $1$.
See \cref{fig:alg} for the circuit diagram.

So, by the Holevo--Helstrom theorem~\cite[Theorem 3.4]{Wat18}, the success of the algorithm can be upper bounded by
\begin{align}
    \Pr[\alg \text{ succeeds}] \leq \frac12 + \frac14\trnorm[\Big]{
        \E_{U \sim \text{Ensemble }0}\bracks[\big]{\proj{\alg_U}}
        - \E_{U \sim \text{Ensemble }1}\bracks[\big]{\proj{\alg_U}}
    }. \label{eq:success}
\end{align}
Recall that both ensembles are over diagonal unitaries $U$, and their diagonal entries $U_i$ are $q$-th roots of unity.
So, we can describe a unitary from the ensemble as a list of group elements of $\cyc_q$.
Let $p^{(0)}(U)$ and $p^{(1)}(U)$ be the probability that the unitary $U$ is chosen from Ensembles 0 and 1, respectively.
Then, using the language of purification, notice that we can define
\begin{align*}
    \ket{\alg^{(0)}}_{\reg{PRS}} &\coloneqq \sum_{U} \sqrt{p^{(0)}(U)} \ket{\alg_U}_{\reg{RS}}\ket{U}_{\reg{P}}, \\
    \ket{\alg^{(1)}}_{\reg{PRS}} &\coloneqq \sum_{U} \sqrt{p^{(1)}(U)} \ket{\alg_U}_{\reg{RS}}\ket{U}_{\reg{P}},
\end{align*}
where $\ket{U} = \ket{U_1} \otimes \dots \otimes \ket{U_d}$. In this way, we have that
\begin{align} \label{eq:purifying-ensemble}
    \trnorm[\Big]{\E_{U \sim \text{Ensemble }0}\bracks[\big]{\proj{\alg_U}} - \E_{U \sim \text{Ensemble }1}\bracks[\big]{\proj{\alg_U}}}
    = \trnorm[\Big]{\tr_\reg{P}[\proj{\alg^{(0)}}] - \tr_\reg{P}[\proj{\alg^{(1)}}]}.
\end{align}

We begin by writing the output state as a sum over Feynman paths.
That is, we write
\begin{align*}
    U = \sum_{x \in [d]} U_x \proj{x}.
\end{align*}
Then, we can decompose 
\begin{align*}
    \ket{\alg_U}_{\reg{RS}}
    &= U_\reg{R} A^{(n)}_{\reg{RS}} \dots U_{\reg{R}} A^{(2)}_{\reg{RS}} U_{\reg{R}} A^{(1)}_{\reg{RS}} \ket{0}_{\reg{RS}} \\
    &= \sum_{x_n,\dots, x_1 \in [d]} U_{x_n} \proj{x_n} A^{(n)} \dots U_{x_2} \proj{x_2} A^{(2)} U_{x_1} \proj{x_1} A^{(1)} \ket{0} \\
    &= \sum_{x_n,\dots, x_1 \in [d]} U_{x_n} \dots U_{x_1} \underbrace{\proj{x_n} A^{(n)} \dots \proj{x_2} A^{(2)} \proj{x_1} A^{(1)} \ket{0}}_{\coloneqq \ket{\feyn_{x_n,\dots,x_1}}} \\
    &= \sum_{x_n,\dots, x_1 \in [d]} U_{x_n} \dots U_{x_1} \ket{\feyn_{x_n,\dots,x_1}}. 
\end{align*}
Notice that $\ket{\feyn_{x_n,\dots,x_1}}$ does not depend on $U$; all that we used was that $U$ is diagonal.

To describe the output state, randomized over $U$ pulled from an ensemble, we consider its purification.
Below, we consider Ensemble 1.
The derivation for Ensemble 0 is identical, except with $\eps$ set to zero:
\begin{align*}
    & \qquad \sum_U \sqrt{p^{(1)}(U)} \ket{\alg_U}_{\reg{RS}} \ket{U}_{\reg{P}} \\
    &= \sum_U \sqrt{p^{(1)}(U)} \sum_{x_n,\dots, x_1 \in [d]} U_{x_n} \dots U_{x_1} \ket{\feyn_{x_n,\dots,x_1}}_{\reg{RS}} \ket{U}_{\reg{P}} \\
    &= \sum_{x_n,\dots, x_1 \in [d]} \ket{\feyn_{x_n,\dots,x_1}}_{\reg{RS}} \sum_U \sqrt{p^{(1)}(U)} U_{x_n} \dots U_{x_1} \ket{U}_{\reg{P}}.\\
\intertext{Now, we expand the purification register.
The entries of $U$, $U_x$, are sampled independently from $\mu_\eps$.
}
    &= \sum_{x_n,\dots, x_1 \in [d]} \ket{\feyn_{x_n,\dots,x_1}}_{\reg{RS}} \otimes \parens[\Big]{
        \bigotimes_{i=1}^{d} \parens[\big]{\sum_{g \in \cyc_q} \sqrt{\mu_\eps(g)} g^{x(i)} \ket{g}}
    }_{\reg{P}}.\\
\intertext{
Above, we use $x(i)$ to denote the number of times that $i$ appears in the multiset $\braces{x_n,\ldots,x_1}$.
At this point, the purification register only depends on the histogram of $x_n,\dots,x_1$: in particular, let $h_i = x(i)$ for $i \in [d]$.
So, we can regroup the sum to be in terms of $h_1, \dots, h_d$, and define $\ket{\feyn_{h_1,\dots,h_d}}$ to be equal to the sum of $\ket{\feyn_{x_n,\dots,x_1}}$ over all $x$'s whose statistics are the same as $h$, meaning that $x(i) = h_i$ for all $i \in [d]$.
Notice that $h_1 + \dots + h_d = n$ always.
}
    &= \sum_{\substack{h_1,\dots,h_d \geq 0 \\ h_1 + \dots + h_d = n}} \ket{\feyn_{h_1,\dots,h_d}}_{\reg{RS}} \otimes \parens[\Big]{
        \bigotimes_{i=1}^{d} \parens[\big]{\sum_{g \in \cyc_q} \sqrt{\mu_\eps(g)} g^{h_i} \ket{g}}
    }_{\reg{P}}. \\
\intertext{
Then, we apply a unitary to the purification register.
Crucially, this does not change the state attained by tracing out the purification register.
We apply a biased Fourier transform, using \cref{lem:biased-ft} to get the final form of the purification.
To use \cref{lem:biased-ft}, we need that $h_1, \dots, h_d \leq n \leq q$.
}
    &\xrightarrow[]{(F^{(\eps)})^{\otimes d}} \sum_{\substack{h_1,\dots,h_d \geq 0 \\ h_1 + \dots + h_d = n}} \ket{\feyn_{h_1,\dots,h_d}}_{\reg{RS}} \otimes \parens[\Big]{
        \bigotimes_{i=1}^{d} \parens[\big]{\alpha_{h_i}^{(i)} \ket{h_i}+ \sum_{k = 0}^{h_i - 1} \alpha_k^{(i)} \ket{k}}
    }_{\reg{P}}.
\end{align*}
Call this final state $\ket{\widetilde{\alg}^{(1)}}_{\reg{PRS}}$.
By \cref{lem:biased-ft}, for all $i$, the $\alpha^{(i)}$'s have an $\ell^2$ norm of one and $\abs{\alpha_{h_i}^{(i)}}^2 \geq 1 - 4\eps^2$.
Further, when $h_i = 0$, then $\abs{\alpha_{h_i}^{(i)}} = 1$ by the aforementioned properties.
So, the coefficient associated to $\bigotimes_{i=1}^{d} \ket{h_i}$, $\alpha_{h_{1}}^{(1)} \dots \alpha_{h_d}^{(d)}$, has size at least
\begin{align} \label{eq:n-eps-squared}
    \prod_{i=1}^{d} \abs{\alpha_{h_i}^{(i)}}^2 \geq (1 - 4\eps^2)^n \geq 1 - 4n\eps^2.
\end{align}
Consequently, if we perform the same series of operations on the unbiased ensemble (essentially setting $\eps = 0$), we get the purification
\begin{align*}
    \ket{\widetilde{\alg}^{(0)}} \coloneqq \sum_{\substack{h_1,\dots,h_d \geq 0 \\ h_1 + \dots + h_d = n}} \ket{\feyn_{h_1,\dots,h_d}}_{\reg{RS}} \otimes \parens[\Big]{
        \bigotimes_{i=1}^{d} \ket{h_i}
    }_{\reg{P}}.
\end{align*}
Now, we return to \eqref{eq:success}:
\begin{align*}
    \Pr[\alg \text{ succeeds}]
    &\leq \frac12 + \frac14\trnorm[\Big]{
        \E_{U \sim \text{Ensemble }0}\bracks[\big]{\proj{\alg_U}}
        - \E_{U \sim \text{Ensemble }1}\bracks[\big]{\proj{\alg_U}}
    } \\
    &= \frac12 + \frac14\trnorm[\Big]{\tr_\reg{P}[\proj{\alg^{(0)}}] - \tr_\reg{P}[\proj{\alg^{(1)}}]} \tag*{by \eqref{eq:purifying-ensemble}}\\
    &= \frac12 + \frac14\trnorm[\Big]{\tr_\reg{P}[\proj{\widetilde{\alg}^{(0)}}] - \tr_\reg{P}[\proj{\widetilde{\alg}^{(1)}}]} \\
    &\leq \frac12 + \frac12(4n\eps^2). \tag*{by \cref{lem:eps-squared} and \eqref{eq:n-eps-squared}}
\end{align*}
In the final step, we apply \cref{lem:eps-squared} with $\ket{\psi} \gets \ket{\widetilde{\alg}^{(0)}}$ and $\ket{\widetilde{\psi}} \gets \ket{\widetilde{\alg}^{(1)}}$.
In particular, we have that the $\ket{\phi_i}$'s are $\ket{\feyn_{h_1,\dots,h_d}}$ and the $\ket{\varphi_i}$'s are $\bigotimes_{i=1}^{d} \ket{h_i}$.
The main assumption of the lemma to check is that the error in $\ket{\widetilde{\psi}}$ is orthogonal to the $\ket{\varphi_i}$'s.
This is true because the error is over computational basis vectors $\ket{y_1} \dots \ket{y_d}$ which sum to smaller than $n$, $y_1 + \dots + y_d < n$, whereas we know that $h_1 + \dots + h_d = n$.

Returning to the equation, (provided that $q \geq n$), the algorithm can succeed with probability $\geq \frac23$ only when $n\eps^2 \geq \frac{1}{12}$.
So, because $q \geq 1/\eps^2$ by assumption, this completes the proof.
\end{proof}

\begin{remark}[What is different with access to $U^\dagger$?]
When an algorithm gets the ability to apply both $U$ and $U^\dagger$, we can analyze the success probability of the algorithm in an identical fashion.
However, we do not get the $n = \bigOmega{1/\eps^2}$ lower bound.
Just looking at Ensemble 0, the purification proceeds identically, with the only change being that some of the $U_x$'s are replaced by $U_x^*$: these inverse queries \emph{decrease} the index in the purification instead of increase.
As a result, the $\ket{h_1}\dots\ket{h_d}$ which appear in the purification need not sum to $n$; and so the difference between Ensembles 0 and 1 are not orthogonal in the manner which allows us to apply \cref{lem:eps-squared}.
In other words, the difference between Ensembles 0 and 1 cannot be captured by probability mass, and so it scales with $\eps$, not $\eps^2$.
\end{remark}

\begin{remark}[Proving the lower bound against $\controlled U$] \label{rmk:ctrl}
    The lower bounds of \cref{thm:lower-ae} and \cref{thm:lower-aa} both continue to hold if we get access to the controlled unitary $\controlled U$, with one modification to the ensembles.
    By \cite[Theorem 1.1]{tw25b}, if $\ket{\alg_{\controlled U}}$ is the output of a circuit which applies $\controlled U$ $n$ times for $n \leq q$, then there is a circuit which applies $U$ $n$ times and outputs $\ket{\alg_{\controlled (\varphi U)}}$ for a random $\varphi \in \cyc_q$.
    When $U$ is drawn from an ensemble which is invariant under multiplication by $\varphi \in \cyc_q$, then the output (when treated as a mixed state over the randomness of $U$) is identical to the output of the original ensemble.
    \begin{align*}
        \E_{U}\bracks[\big]{\proj{\alg_{\controlled U}}}
        = \E_{U}\bracks[\big]{\E_{\varphi \in \cyc_q}\bracks[\big]{\proj{\alg_{\controlled (\varphi U)}}}}
    \end{align*}
    So, it suffices to modify the ensembles to make them invariant to global phase: we change Ensemble 0 and 1 of \cref{def:ae-ensembles} to be $\varphi U$, where $U$ is drawn from the original definition of Ensemble 0 and 1, and $\varphi$ is sampled from $\mu_0$.
    Because this global phase does not affect circuits which use uncontrolled queries, these new ensembles can be distinguished using amplitude estimation, by an identical argument.
    With this random global phase, the ensembles also become hard to distinguish given controlled queries.
\end{remark}

\subsection{Proof of \texorpdfstring{\cref{lem:biased-ft}}{Lemma 3.3}}

For each $k \in \{0, \ldots, q-1\}$, define the vectors
\begin{equation*}
    \ket{F_k} \coloneqq \frac{1}{\sqrt{q}} \sum_{g \in \cyc_q} g^k \ket{g},
    \quad\text{and}\quad
    \ket{\widetilde{F}_k} \coloneqq \sum_{g \in \cyc_q} \sqrt{\mu_{\epsilon}(g)} g^k \ket{g}.
\end{equation*}
The vectors $\ket{F_0}, \ldots, \ket{F_{q-1}}$ form the usual Fourier basis,
whereas the vectors $\ket{\widetilde{F}_0}, \ldots, \ket{\widetilde{F}_{q-1}}$ can be viewed as a biased version of the Fourier basis.
These vectors are unit vectors, because
\begin{equation*}
\norm{\ket{\widetilde{F}_k}}^2
= \braket{\widetilde{F}_k}{\widetilde{F}_k}
= \Big(\sum_{g \in \cyc_q} \sqrt{\mu_\eps(g)} g^{-k} \bra{g}\Big) \Big(\sum_{g \in \cyc_q} \sqrt{\mu_\eps(g)} g^k \ket{g}\Big)
= \sum_{g \in \cyc_q} \mu_\eps(g)
= 1.
\end{equation*}
However, they are \emph{not} orthogonal,
which means that there can be no unitary which exactly maps $\ket{\widetilde{F}_k}$ to $\ket{k}$ for all $k$.
\Cref{lem:biased-ft} states that there is still a unitary which approximately performs this mapping.
To show this, we will first show that although these vectors are not exactly orthogonal, they are still \emph{approximately} orthogonal.
Define the matrix
\begin{equation}
    \widetilde{F} \coloneqq \sum_{k \in \Z_q} \ketbra{\widetilde{F}_k}{k}
    = \begin{bmatrix}
        \vert & \vert & \cdots & \vert \\
        \ket{\widetilde{F}_0} & \ket{\widetilde{F}_1} & \cdots & \ket{\widetilde{F}_{q-1}} \\
        \vert & \vert & \cdots & \vert
    \end{bmatrix}. \label{eq:fake-SVD}
\end{equation}
If the $\ket{\widetilde{F}_k}$ vectors were orthogonal, then \Cref{eq:fake-SVD} would give the singular value decomposition of $\widetilde{F}$, and therefore all of $\widetilde{F}$'s singular values would be equal to 1.
This is of course not true,
but the next lemma nevertheless shows that $\widetilde{F}$'s singular values are close to 1.
\begin{lemma}[Singular values of $\widetilde{F}$]\label{lem:ms-svd}
    $\sqrt{1- \epsilon} \leq \sigma_{\min}(\widetilde{F}) \leq \sigma_{\max}(\widetilde{F}) \leq \sqrt{1 + 2\epsilon}$.
\end{lemma}
\begin{proof}
We directly calculate the singular value decomposition of $\widetilde{F}$ as follows:
\begin{align*}
    \widetilde{F}
    = \sum_{k \in \Z_q} \ketbra{\widetilde{F}_k}{k}
    &= \sum_{k \in \Z_q} \Big(\sum_{g \in \cyc_q} \sqrt{\mu_{\epsilon}(g)} g^k \ket{g}\Big)\bra{k}\\
    &= \sum_{g \in \cyc_q} \sqrt{q\cdot\mu_{\epsilon}(g)}\ket{g} \Big( \frac{1}{\sqrt{q}}\sum_{k \in \Z_q} g^k \bra{k}\Big)
    = \sum_{g \in \cyc_q} \sqrt{q\cdot \mu_{\epsilon}(g)}\cdot \ketbra{g}{F_{g^{-1}}}.
\end{align*}
This implies that the singular values of $\widetilde{F}$ are $\sqrt{q \cdot \mu_{\epsilon}(g)}$ for $g \in \cyc_q$;
as $(1 - \eps)/q \leq \mu_{\epsilon}(g) \leq (1 + 2\epsilon)/q$ for all $g$,
this means that the singular values of $\widetilde{F}$ are in the interval $(\sqrt{1 - \epsilon}, \sqrt{1+2\eps})$.
\end{proof}

Now we will perform the Gram-Schmidt process on the columns of $\widetilde{F}$ to ``round'' them into an orthonormal set of vectors. For each $0 \leq k \leq q-1$, write
\begin{equation*}
    \Pi_k = \text{the projector onto }\mathrm{span}\{\ket{\widetilde{F}_0}, \ldots, \ket{\widetilde{F}_k}\}.
\end{equation*}
The Gram-Schmidt process computes the set of unnormalized vectors $\ket{G_0}, \ldots, \ket{G_{q-1}}$ given by
\begin{equation*}
    \ket{G_k}
    \coloneqq \overline{\Pi}_{k-1} \cdot \ket{\widetilde{F}_k}
    = \ket{\widetilde{F}_k} - \Pi_{k-1} \cdot \ket{\widetilde{F}_k}.
\end{equation*}
It outputs the normalized vectors
$\ket{\overline{G}_0}, \ldots, \ket{\overline{G}_{q-1}}$
in which $\ket{\overline{G}_k} = \ket{G_k} / \norm{\ket{G_k}}$ for all $k \in \{0, \ldots, q-1\}$.
These vectors form an orthonormal basis,
because if $j < k$, then $\ket{\overline{G}_j}$ is contained in $\Pi_{k-1}$ and $\ket{\overline{G}_k}$ is orthogonal to $\Pi_{k-1}$.
We will aim to show that each $\ket{G_k}$ has length very close to 1,
and hence must be close to $\ket{\widetilde{F}_k}$.
Let us fix some $k \in \{0, \ldots, q-1\}$. 
Then
\begin{equation}\label{eq:length-of-G}
    \norm{\ket{G_k}}^2
    = \braket{G_k}{G_k}
    = \bra{\widetilde{F}_k} \cdot \overline{\Pi}_{k-1} \cdot \ket{\widetilde{F}_k}
    = \braket{\widetilde{F}_k}{\widetilde{F}_k}
     - \bra{\widetilde{F}_k} \cdot \Pi_{k-1} \cdot \ket{\widetilde{F}_k}
     = 1
     - \bra{\widetilde{F}_k} \cdot \Pi_{k-1} \cdot \ket{\widetilde{F}_k}.
\end{equation}
Our goal is therefore to upper-bound
$\bra{\widetilde{F}_k} \cdot \Pi_{k-1} \cdot \ket{\widetilde{F}_k}$. The next lemma shows
our first step in upper bounding this quantity.
\begin{lemma}\label{lem:cant-think-of-a-good-name}
    For any vector $\ket{v} \in \C^q$,
    \begin{equation*}
        \bra{v} \cdot \Pi_{k-1} \cdot \ket{v}
        \leq \frac{1}{1 - \epsilon} \cdot \sum_{j=0}^{k-1} |\braket{\widetilde{F}_j}{v}|^2.
    \end{equation*}
\end{lemma}
\begin{proof}
By \cref{lem:ms-svd}, the singular values of $\widetilde{F}$ are at least $\sqrt{1-\eps}$.
So, if we take a subset of columns of $\widetilde{F}$,
\begin{align*}
    M_k \coloneqq \sum_{j = 0}^{k-1}  \ketbra{\widetilde{F}_j}{j}
    = \begin{bmatrix}
        \vert & \vert & \cdots & \vert \\
        \ket{\widetilde{F}_0} & \ket{\widetilde{F}_1} & \cdots & \ket{\widetilde{F}_{k-1}} \\
        \vert & \vert & \cdots & \vert
    \end{bmatrix},
\end{align*}
the singular values of this matrix will also be at least $\sqrt{1-\eps}$.
Let $M_k = \sum_{j=0}^{k-1} \sigma_i \ketbra{u_i}{w_i}$ be a singular value decomposition of $M_k$.
Then $\Pi_{k-1} = \sum_{j=0}^{k-1} \ketbra{u_i}{u_i}$, so
\begin{equation*}
    \bra{v} \cdot \Pi_{k-1} \cdot \ket{v}
    = \sum_{j=0}^{k-1} \braket{v}{u_i}\braket{u_i}{v}
    \leq \sum_{j=0}^{k-1} \frac{\sigma_i^2}{1-\eps}\braket{v}{u_i}\braket{u_i}{v}
    = \frac{1}{1-\eps}\bra{v} \cdot M_k \cdot M_k^\dagger \cdot \ket{v}
    = \frac{1}{1-\eps} \cdot \sum_{j=0}^{k-1} \abs{\braket{\widetilde{F}_j}{v}}^2. \qedhere
\end{equation*}
\end{proof}
With \Cref{lem:cant-think-of-a-good-name} in hand,
we now prove the following upper bound on $\bra{\widetilde{F}_k} \cdot \Pi_{k-1} \cdot\ket{\widetilde{F}_k}$.
\begin{lemma}\label{lem:outer-product}
    $\bra{\widetilde{F}_k} \cdot \Pi_{k-1} \cdot\ket{\widetilde{F}_k} \leq 2\epsilon^2/(1-\epsilon)$.
\end{lemma}
\begin{proof}
    By \Cref{lem:cant-think-of-a-good-name},
    \begin{equation*}
    \bra{\widetilde{F}_k} \cdot \Pi_{k-1} \cdot\ket{\widetilde{F}_k}
    \leq \frac{1}{1 - \epsilon}\cdot\sum_{j=0}^{k-1} |\braket{\widetilde{F}_j}{\widetilde{F}_k}|^2
    \leq \frac{1}{1 - \epsilon}\cdot\sum_{j \neq k} |\braket{\widetilde{F}_j}{\widetilde{F}_k}|^2.
    \end{equation*}
    Substituting in the definition of the $\ket{\widetilde{F}_j}$ vectors, this is equal to
    \begin{equation}\label{eq:step-two}
        \sum_{j \neq k} |\braket{\widetilde{F}_j}{\widetilde{F}_k}|^2 = \sum_{j\neq k} \Big| \E_{g \sim \mu_{\eps}} g^{k-j} \Big|^2
        = \sum_{i=1}^{q-1} \Big| \E_{g \sim \mu_{\eps}} g^{i} \Big|^2
        = \sum_{i=1}^{q-1} \Big| (1-\eps)\E_{g \sim \mu_0} \bracks{g^{i}} + \eps \E_{g \sim \mu_1} \bracks{g^{i}} \Big|^2
        = \eps^2 \cdot \sum_{i=1}^{q-1} \Big| \E_{g \sim \mu_1} g^{i} \Big|^2.
    \end{equation}
    The last equality uses that $\E_{g \sim \mu_0} g^i = 0$ unless $i = 0$ (mod $q$).
    \begin{multline}\label{eq:step-three}
    \eqref{eq:step-two}
    = \eps^2 \cdot \sum_{i=1}^{q-1} \Big(\E_{g \sim \mu_1} g^{i}\Big) \Big(\E_{h \sim \mu_1} h^{-i}\Big)
    = \eps^2 \cdot \E_{g, h \sim \mu_1} \sum_{i=1}^{q-1} (g h^{-1})^{i} \\
    = \eps^2 \cdot \parens[\Big]{\E_{g, h \sim \mu_1} \sum_{i=0}^{q-1} (g h^{-1})^{i} - 1}
    = \eps^2 \cdot \parens[\Big]{q \cdot \Pr_{g,h \sim \mu_1}[g = h] - 1}
    \end{multline}
    Here, the last equality uses that the expression $\sum_{i=0}^{q-1} (gh^{-1})^i$ is equal to $q$ if $gh^{-1} = 1$ and is 0 for all other $gh^{-1} \in \cyc_q$.
    Since $\mu_1$ is the uniform distribution over $2\floor{q/4} + 1$ elements of $\cyc_q$, the probability that two elements $g, h$ sampled from $\mu_1$ are equal is $1/(2\floor{q/4} + 1)$.
    Hence,
    \begin{align*}
        \eqref{eq:step-three}
        = \eps^2 \cdot \parens[\Big]{\frac{q}{2\floor{q/4} + 1} - 1}
        \leq 2 \epsilon^2.
    \end{align*}
    Plugging this in above completes the proof.
\end{proof}

Combined with \Cref{eq:length-of-G},
\Cref{lem:outer-product} implies that
$\braket{G_k}{G_k} \geq 1 - 2\epsilon^2/(1-\epsilon)$.
Now we would like to conclude by showing that $\ket{\overline{G}_k}$ has large overlap with $\ket{\widetilde{F}_k}$. To begin,
\begin{equation*}
    \braket{G_k}{\widetilde{F}_k}
    = \bra{G_k} \cdot \overline{\Pi}_{k-1} \cdot \ket{\widetilde{F}_k}
    = \braket{G_k}{G_k}
    = \norm{\ket{G_k}}^2.
\end{equation*}
Therefore,
\begin{equation*}
    \braket{\overline{G}_k}{\widetilde{F}_k}
    = \braket{G_k}{\widetilde{F}_k} / \norm{\ket{G_k}}
    = \norm{ \ket{G_k} }^2 / \norm{\ket{G_k}}
    = \norm{ \ket{G_k} } \geq \sqrt{1 - 2\epsilon^2/(1-\epsilon)}.
\end{equation*}
We now use this to prove \Cref{lem:biased-ft}.

\begin{proof}[Proof of \Cref{lem:biased-ft}]
We define the unitary matrix $F^{(\eps)}$ as follows:
\begin{equation*}
    F^{(\eps)}
    = \sum_{j \in \Z_q} \ketbra{j}{\overline{G}_j}.
\end{equation*}
Applying this to $\ket{\widetilde{F}_k}$, we get
\begin{equation*}
    F^{(\eps)}\cdot \ket{\widetilde{F}_k}
    = \Big(\sum_{j \in \Z_q} \ketbra{j}{\overline{G}_j}\Big)\cdot \ket{\widetilde{F}_k}
    = \sum_{j \in \Z_q} \alpha_j \ket{j},
\end{equation*}
where $\alpha_j = \braket{\overline{G}_j}{\widetilde{F}_k}$.
Because $F^{(\epsilon)}$ is unitary, $\sum_{j \in \Z_q} |\alpha_j|^2 = 1$.
Furthermore, $\alpha_j = 0$ for all $j > k$ because $\ket{\widetilde{F}_k}$ is contained in $\Pi_{j-1}$ and $\ket{\overline{G}_j}$ is orthogonal to $\Pi_{j-1}$.
Finally,
    \begin{equation*}
        \alpha_k = \braket{\overline{G}_k}{\widetilde{F}_k}
    \geq \sqrt{1 - 2\epsilon^2/(1-\epsilon)}.
    \end{equation*}
This completes the proof.
\end{proof}

\section{Extending the lower bound to amplitude amplification}

Now, we prove \cref{thm:lower-aa-main}.
Our proof builds off of the amplitude estimation lower bound: we adapt the hard instance to the amplification setting, and then lift the estimation lower bound to a lower bound for the adapted instance.

\begin{problem}[Distinguishing biased unitary ensembles] \label{prob:aa-distinguish}
    For a dimension parameter $d$, a bias parameter $\eps \in (0,1)$, and an order parameter $q$, consider the following problem.
    Let $b \in \braces{1,2}$ be a bit which is 1 or 2 with probability $1/2$.
    Let $V$ be drawn from ensemble 1 of \cref{def:ae-ensembles}.
    Then draw a unitary $U$ from ensemble $b$:
    \begin{enumerate}
        \item If $b = 1$, then $U = V$;
        \item If $b = 2$, then $U = DV$ where $D$ is a diagonal unitary with entries $\bra{k}D\ket{k} = e^{2\pi \ii k / d}$.
    \end{enumerate}
    The goal is to, given quantum access to the unitary $U$, output a $\wh{b} \in \braces{1,2}$ such that $\wh{b} = b$ with good probability.
\end{problem}

This distinguishing task can be solved with amplitude amplification.

\begin{proposition} \label{prop:aa-reduction}
    Consider the distinguishing task \cref{prob:aa-distinguish} with parameters $\eps \in (0,1)$, $d \geq C/\eps^2$, and $q \geq C$ with $C$ a sufficiently large constant.
    Then there is an algorithm which succeeds at this task with probability $\geq 0.9$, using one call to amplitude amplification (\cref{def:aa-ae}).
\end{proposition}
\begin{proof}
Let $T \in \C^{d \times d}$ be a unitary matrix whose first two columns are as follows:
\begin{align*}
    T \ket{0} &= \frac{1}{\sqrt{d}} \sum_{k=0}^{d-1} \ket{k}; \\
    T \ket{1} &= DT\ket{0} = \frac{1}{\sqrt{d}} \sum_{k=0}^{d-1} e^{2\pi \ii k / d} \ket{k}.
\end{align*}
Further, let $Z \in \C^{2d \times 2d}$ be a unitary matrix which ``marks'' $\ket{0}$ and $\ket{1}$: $Z \ket{0} \ket{j} = \ket{0}\ket{j \oplus 1}$, $Z \ket{1} \ket{j} = \ket{1}\ket{j \oplus 1}$, and $Z \ket{i}\ket{j} = \ket{i}\ket{j}$ for $i \neq 0,1$.
Then the quantum circuit
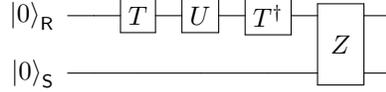
\begin{align*}
    \Qcircuit @R=1em @C=1em {
    \lstick{\ket{0}_{\reg{R}}} & \qw & \gate{T} & \gate{U} & \gate{T^\dagger} & \multigate{1}{Z} & \qw \\
    \lstick{\ket{0}_{\reg{S}}} & \qw & \qw & \qw & \qw & \ghost{Z} & \qw
    }
\end{align*}
outputs the state
\begin{align*}
    &\sqrt{\abs{\alpha}^2 + \abs{\beta}^2}\ket{\phi_{\good}}_{\reg{R}} \ket{1}_{\reg{S}} + \sqrt{1 - \abs{\alpha}^2 - \abs{\beta}^2} \ket{\phi_{\bad}}_{\reg{R}} \ket{0}_{\reg{S}} \\
    & \qquad\qquad \text{where } \alpha \coloneqq \ntr(U),\, \beta \coloneqq \ntr(D^\dagger U) \text{, and } \ket{\phi_{\good}} = \frac{\alpha\ket{0} + \beta \ket{1}}{\sqrt{\abs{\alpha}^2 + \abs{\beta}^2}}.
\end{align*}
If we call amplitude amplification (\cref{def:aa-ae}) on this circuit, $Z_{\reg{RS}} T^\dagger_{\reg{R}} U_{\reg{R}} T_{\reg{R}}$, we get a state $\rho$ such that $\frac12 \trnorm{\rho - \proj{\phi_\good}} < 0.01$.
Consider measuring $\rho$ in the computational basis.
By \cref{lem:traces-concentrate}, for $d \geq \frac{C}{\eps^2}$ for sufficiently large $C$,
\begin{align}
    \Pr\bracks[\Big]{\abs[\Big]{\alpha - \E[\alpha]} \geq 0.01 \eps} &\leq 0.01;
    \label{eq:alpha-concentrates} \\
    \Pr\bracks[\Big]{\abs[\Big]{\beta - \E[\beta]} \geq 0.01 \eps} &\leq 0.01.
    \label{eq:beta-concentrates}
\end{align}
By \cref{lem:circle-means}, when $b = 1$, $\frac{\eps}{2} < \E[\alpha] < \eps$ and $\E[\beta] = 0$; when $b = 2$, $\E[\alpha] = 0$ and $\frac{\eps}{2} < \E[\beta] < \eps$.
So, with probability $\geq 0.98$, when $b = 1$,
\begin{align*}
    \frac{\abs{\alpha}^2}{\abs{\alpha}^2 + \abs{\beta}^2} = \frac{1}{1 + \abs{\beta/\alpha}^2} > \frac{1}{1 + (0.01/0.49)^2} > 0.99.
\end{align*}
So, if we measure $\rho$ in the computational basis, we see the outcome $\ket{0}$ with probability $\geq 0.95$.
On the other hand, reproducing the argument for $b = 2$, for that ensemble, if we measure the output of amplitude amplification, $\rho$, we see the outcome $\ket{1}$ with probability $\geq 0.95$.

This leads to the following distinguishing algorithm.
Run amplitude amplification on the specified circuit, and measure the resulting state in the computational basis: if the outcome is $\ket{0}$, output $\wh{b} = 1$, and if the outcome is $\ket{1}$, output $\wh{b} = 2$.
This algorithm succeeds with probability $\geq 0.9$.
\end{proof}

Our lower bound in \cref{thm:lower-ae} can be lifted to prove a lower bound for this problem.

\begin{theorem} \label{thm:lower-aa}
    Consider \cref{prob:aa-distinguish} with $\eps \in (0, 1/2)$ and $q \geq \frac{1}{\eps^2}$.
    Then any algorithm to solve it with probability $\geq \frac56$ requires $n \geq \frac{1}{12\eps^2}$ applications of $U$.
\end{theorem}
\begin{proof}[Proof of \cref{thm:lower-aa}]
Suppose we had a quantum circuit which applies $U$ $n < \frac{1}{12\eps^2}$ times and solves \cref{prob:aa-distinguish} with probability $\geq \frac56$.
Then we can use this algorithm to solve \cref{prob:ae-distinguish}.

Let $U$ be a unitary drawn from either ensemble 0 or ensemble 1 of \cref{prob:ae-distinguish} (with $a = 0$).
Then sample a coin $b' \in \braces{1,2}$ with equal probability.
If $b' = 1$, then run the distinguishing algorithm on $U$; if $b' = 2$, then run the distinguishing algorithm on $D U$.
Because of the guarantee, when $U$ is drawn from ensemble 1, the output $\wh{b} = b'$ with probability at least $\frac{5}{6}$.
When $U$ is drawn from ensemble 0, the output $\wh{b} = b'$ with probability exactly $\frac12$, since in this case, $U$ and $DU$ follow identical distributions.

So, we can distinguish ensemble 0 from 1 by outputting the guess $1$ if $\wh{b} = b'$, and $0$ otherwise.
This algorithm succeeds with probability $\geq \frac12(\frac12 + \frac56) = \frac23$, contradicting \cref{thm:lower-ae}.
\end{proof}

Finally, with this, we can prove our main theorem as a corollary.

\begin{proof}[Proof of \cref{thm:lower-aa-main}]
We begin by discussing a subtlety in modeling algorithms for amplitude amplification which did not arise in the proof of our amplitude estimation lower bound. Given a unitary matrix $X$ such that
    \begin{align}\label{eq:this-is-a-repeat}
        X\ket{0} = a\ket{\phi_\good}\ket{0} + \sqrt{1 - a^2}\ket{\phi_\bad} \ket{1},
    \end{align}
algorithms for amplitude amplification, such as the original algorithm due to Brassard, H{\o}yer, Mosca, and Tapp~\cite{BHMT02}, are allowed to apply the unitary $X$ a number of times which depends on the amplitude parameter $a$ and increases to infinity as $a$ tends towards 0. This means that these algorithms cannot be modeled as a circuit of a fixed size, but rather as an outer classical loop which grows the quantum circuit until the loop terminates.
This is an issue because our purification-based lower bound techniques assume that the algorithm is formatted as a fixed quantum circuit.

To handle this issue, we structure our proof as a proof by contradiction.
Suppose for the sake of contradiction that there is an algorithm $\calA$ which solves the amplitude amplification problem with $\littleO{\min(d, 1/a^2)}$ applications of $X$.
Then for any lower bound $a^* > 0$,
if $\calA$ is provided a unitary $X$ satisfying \eqref{eq:this-is-a-repeat} with $a \geq a^*$,
it will perform amplitude amplification using at most a finite number $n = \littleO{\min(d, 1/(a^*)^2)}$ of applications of $X$.
This algorithm can then be compiled into a fixed quantum circuit $\calC$
which performs $n$ applications of $X$ and behaves identically to $\calA$ on all unitaries $X$ of the form \eqref{eq:this-is-a-repeat} with $a \geq a^*$.
In this way, we can assume that our algorithm is modeled as a fixed quantum circuit,
and from here, the proof proceeds mostly along the lines of the proof of \cref{thm:lower-ae-main}.
Ultimately, we will show that $n \geq \bigO{1} \cdot \min(d, 1/(a^*)^2)$ applications of $X$ are necessary to perform amplitude amplification when promised that $a \geq a^*$, which leads to a contradiction.

Let $C$ be the sufficiently large constant from \cref{prop:aa-reduction}.
We will assume that $a^* \in (0,0.125)$ and that $d > 4C$, as the lower bound is vacuously true when these bounds do not hold.
Let
\begin{equation*}
    \epsilon = \max\{4a^*, \sqrt{C/d}\},
\end{equation*}
and note that $\epsilon \in (0, 1/2)$ by our assumptions on $a$ and $d$.
Consider the distinguishing task from \cref{prob:aa-distinguish} with bias parameter $\epsilon$
and order parameter
$q = \lceil\max\{1/\eps^2, C\}\rceil$.
As $d \geq C/\epsilon^2$ and $q \geq C$, we can apply \cref{prop:aa-reduction},
which states that there is an algorithm which succeeds at this distinguishing task with probability $\geq 0.9$ using one call to amplitude amplification.
Inspecting the proof of \cref{prop:aa-reduction},
we see that in the ``probability 0.98 case'', the amplitude amplification subroutine is provided a unitary matrix $X$ as in \eqref{eq:this-is-a-repeat} in which the amplitude on the ``good'' state is
\begin{equation*}
    a = \sqrt{|\alpha|^2 + |\beta|^2}
    \geq \max\{|\alpha|, |\beta|\}
    \geq 0.49 \epsilon
    \geq a^*.
\end{equation*}
Hence, to perform amplitude amplification, we can use the circuit $C$ resulting from compiling $\calA$ using the amplitude lower bound of $a^*$.
But since $\epsilon \in (0, 1/2)$ and $q \geq 1/\epsilon^2$, we can apply \cref{thm:lower-aa}, which states that this algorithm requires
\begin{equation*}
    n \geq \frac{1}{12\epsilon^2}
    = \min\Big(\frac{1}{192 (a^*)^2}, \frac{d}{12C}\Big)
    = \bigO{1} \cdot \min(1/(a^*)^2, d)
\end{equation*}
applications of $U$.
As the only applications of $U$ in the algorithm occur in the amplitude amplification subroutine,
amplitude amplification with error $\epsilon$ requires this many applications of $U$.
This gives the desired contradiction, finishing the proof.
\end{proof}

\section*{Acknowledgments}

E.T.\ is supported by the Miller Institute for Basic Research in Science, University of California Berkeley. 
E.T.\ thanks Robin Kothari, Anirban Chowdhury, and Arjan Cornelissen for discussions.
J.W.\ is supported by the NSF CAREER award CCF-2339711, as well as the Ewin Tang Open Problems Prize Fund.
We thank Gregory Rosenthal for carefully reading an earlier version of this document.

\printbibliography

\end{document}